\documentclass[conference,a4paper]{IEEEtran}

\usepackage{pgf,tikz,pgfplots}

\usepackage[style=ieee,citestyle=numeric-comp,sorting=nty]{biblatex}
\usepackage{amsmath,amssymb,amsfonts,amsthm,mathtools}

\newtheorem{remark}{Remark}

\newtheorem*{probdef*}{Problem Definition}
\newtheorem{claim}{Claim}

\DeclareSymbolFont{bbold}{U}{bbold}{m}{n}
\DeclareSymbolFontAlphabet{\mathbbold}{bbold}

\title{Transaction Confirmation in Coded Blockchain}

\author{%
  \IEEEauthorblockN{\textbf{Ilan Tennenhouse}$^*$ and \textbf{Netanel Raviv}$^\ddag$}
  \IEEEauthorblockA{%
    $^*$Blavatnik School of Computer Science, Tel Aviv University, Tel Aviv, Israel\\
    $^\ddag$Department of Computer Science and Engineering, Washington University in St. Louis, St. Louis, MO, USA}
}

\begin{document}
\maketitle
\pagestyle{plain}

\begin{abstract}
    As blockchains continue to seek to scale to a larger number of nodes, the communication complexity of protocols has become a significant priority as the network can quickly become overburdened. Several schemes have attempted to address this, one of which uses coded computation to lighten the load. Here we seek to address one issue with all such coded blockchain schemes known to the authors: transaction confirmation. In a coded blockchain, only the leader has access to the uncoded block, while the nodes receive encoded data that makes it effectively impossible for them to identify which transactions were included in the block. As a result, a Byzantine leader might choose not to notify a sender or receiver of a transaction that the transaction went into the block, and even with an honest leader, they would not be able to produce a proof of a transaction's inclusion. To address this, we have constructed a protocol to send the nodes enough information so that a client sending or receiving a transaction is guaranteed to not only be notified but also to receive a proof of that transaction's inclusion in the block. Crucially, we do this without substantially increasing the bit complexity of the original coded blockchain protocol.
\end{abstract}

\footnotetext[1]{This work was done while the first author was an undergraduate student in the Department of Computer Science and Engineering in Washington University in St. Louis.}


\section{Introduction}

Blockchains, i.e., append only sequences of blocks of transactions backed by hash chains, have exploded in use over the last decade, both in cryptocurrencies like Bitcoin \cite{Bitcoin} and in enterprise \cite{vmware} for their resistance to tampering. In addition to tampering resistance, blockchains are also lauded for removing the necessity of intermediaries like escrow companies.

Although the terminology differs from one scheme to another, blockchains typically consist of nodes and clients. Each round, the transactions posed by clients are gathered into a block and proposed by one of the nodes to the rest of the nodes who then decide on the validity of the block. Blocks can be considered invalid if a transaction in the block double-spends a single coin/transaction, a signature on one of the transactions is faulty, etc. Nodes therefore have stringent requirements posed on them as they must have the ability to detect invalid transactions which often requires storing all or some significant subset of the chain's history. Note the two groups, nodes and clients, are not mutually exclusive and often there are distinctions made between different types of nodes based on the storage/work/communication requirements imposed on them. After a block has been validated by the nodes, it is appended to the chain, and the next round begins.



Traditional blockchains send every transaction to every node. Denoting by $B$ the block of transactions and by $N$ the number of validator nodes (nodes who determine whether or not the block is valid; i.e. no double spending, etc.) in the network, this incurs a communication bit complexity of $\Omega(N|B|)$, which inhibits the system's throughput and does not scale well with the addition of extra nodes.

One attempt to address this issue has been through \textit{sharding} the blockchain~\cite{shard}, which divides the chain to~$K$ distinct parts called shards, and assigns each node to a single shard. By doing so, nodes only have to communicate with other nodes which hold the same shard (known as a committee) and only require the transactions relevant to their shard. This reduces the bit complexity substantially in cases where~$K$ is proportional to~$N$, but the security of the chain is degraded substantially. Malicious nodes need only to reach a majority within a single committee in order to breach the safety of the chain. In addition, sharding creates further complications for validating and appending transactions between different shards, called \textit{cross-shard} transactions. Much work has been done to address the security issue in sharded blockchain, such as random reassignment of shards~\cite{randomShard}. In this paper we focus our attention on a recent idea which takes sharding one step further. 

Coded blockchain, first proposed in \cite{poly}, is a relatively new idea involving encoding the shards and the block to be verified using a certain error correcting code, and replacing the traditional verification process with the computation of a polynomial over the encoded version of the shard and the block. With the aid of recent advances in coded computation~\cite{LCC}, the results of this polynomial on the uncoded shard and block can be retrieved from sufficiently many polynomial results on encoded ones even if some of the results are incorrect, and hence nodes can determine the validity of the block. The central benefits of coded blockchain are improved security over sharding and reduced communication compared to conventional blockchains. First, since each node holds a linear combination (over a finite field) of blocks, no one committee has exclusive control over any shard, thus eliminating the security problem of sharding. Second, by incorporating information dispersal techniques into the decentralized encoding process, the communication bit complexity can be drastically reduced~\cite{coded_journal}. Further, with a verification polynomial of low enough degree, a coded blockchain can accommodate malicious nodes up to a constant fraction of~$N$. Additional benefits include reduced storage and inherent support for cross-shard transactions.


The reduced bit complexity achieved in~\cite{coded_journal} proceeds in the following steps. First, clients are partitioned into~$K$ distinct communities of equal size. Then, rounds begin by selecting a leader for the round who decides which transactions go into the block and then organizes the block so that all transactions from community $i$ to community $j$ lie in the $(i, j)$'th entry in the block when viewing it as a two-dimensional matrix. To accommodate multiple transactions from the $i$'th community to the $j$'th community, the $(i, j)$'th entry is a collection of transactions, called a \textit{tiny block}. The leader then uses the Lagrange polynomial to encode every row of the block into a single row, and sends a different evaluation to each node. A similar process is done for columns of the block. 

Consequently, each node receives a uniquely encoded row that is a combination of all rows of the block, and a similar uniquely encoded column. The coded column will be appended to the node's history of the chain, while the coded row is used in the validation process to verify against the stored (coded) chain. Regardless, a benefit of encoding the block with the Lagrange polynomial is that Lagrange Coded Computation~\cite{LCC} can be used. Nodes can then compute an arbitrary polynomial over the uncoded block by first computing it over their coded rows and coded histories, and then distribute the results to each other; when a node has enough results it can decode the polynomial verification results. The authors then identified a polynomial that is able to validate the block, and devised respective adversary-resilient encoding and decoding mechanisms with low bit complexity.

Gains in terms of bit complexity, however, come at a price. The mechanism which enables~\cite{coded_journal} to reduce bit complexity (with respect to~\cite{poly}, and any other blockchain design known to the authors) is based on the fact the nodes need only to receive a coded row and a coded column, which constitute a $2/K$ fraction of the block. Secure encoding and decoding processes are orchestrated by the round leader, which relieves the nodes from having to perform the encoding themselves (as in~\cite{poly}).
An undesirable consequence is that no node except the leader has individual information about the uncoded block. Nodes only see one coded row and one coded column, and decoding requires as many as~$\Theta(N)$ of those. 
Consequently, the protocol has no way of confirming to a client whether or not the transaction was in the block without incurring bit complexity that would be equivalent to distributing the uncoded block as-is. Doing so would eliminate the very attractive properties of coded blockchain. 

While this problem originates in~\cite{coded_journal}, it is clear it will persist in future coded blockchain systems---reduced bit complexity can be achieved only if encoding is done jointly using a leader, and certainly not if nodes receive all transactions and perform encoding themselves. It is also clear that this problem does not exist in traditional blockchain system, since every node can attest to every transaction in the chain. Therefore, in order to position coded blockchain as a competitive candidate to resolve communication and security issues in future blockchain systems, a \textit{transaction confirmation} protocol must be devised. In what follows we present such a protocol that can be fitted over any coded blockchain in order to allow clients to interact in a meaningful manner with the chain.


\subsection*{Our Contributions}
We begin in Section~\ref{sec:pre} by surveying the cryptographic primitives and the communication model in the system. These definitions are then used to provide a broad structural definition for coded blockchain. In 
Section~\ref{section:protocol} we construct a protocol which ensures that every sender and receiver of a transaction in the block will be notified of the transaction's inclusion in the block. 
Additionally, the client is provided with a proof that the transaction went into the block, which any node in the network as well as the client receiving the proof can validate. 
Importantly, we do so in bit complexity $O(|B|+H+N\operatorname{polylog}(N))$, where the original coded-blockchain protocol verifies a block by communicating~$O(H)$ bits for some function~$H$ of the system parameters; this is no worse than the bit complexity from~\cite{coded_journal}. 

Proofs of correctness and complexity analysis are given in Section~\ref{section:analysis}, and discussion about additional potential benefits and future work is given in Section~\ref{section:discussion}. Our results reveal that coded blockchains can incorporate uncoded processes (i.e. processes relying on nodes other than the leader having knowledge of the uncoded block) into their protocol without incurring significant cost, and that they in fact should seek to do this when possible. 
\section{Preliminaries}\label{sec:pre}
\subsection{Background}\label{section:background}
We assume a decentralized network of nodes, some of which are malicious, that communicate periodically with clients and with one another over rounds. At each round a leader election mechanism is performed, and the elected leader collects transactions from clients, that are to be verified against the existing data in the system. The system operates in the standard \textit{partial synchrony} model~\cite{partial}, in which message delivery is asynchronous until some future point called Global Stabilization Time (GST), after which the system becomes synchronous, i.e., all messages are delivered in no more than~$\delta$ units of time.

The system is capable of performing a \textit{Partially Synchronous Byzantine broadcast} operation as defined in~\cite{good_case}, where a leader node disseminates a message to all other nodes, which guarantees \textit{agreement}, \textit{validity}, and \textit{termination}, though we additionally require \textit{liveness}. The agreement property guarantees that if a broadcast (that is, a message) is accepted by any honest node, then all honest nodes agree on the value of that broadcast. Validity ensures that if the leader is honest and GST has passed, then the value accepted by any honest node is the value the leader sent out. The termination property requires all honest parties to accept a value after GST, while the liveness property further strengthens this to require that the time for a value to be accepted by all honest nodes is bounded after GST. For instance, the protocol HotStuff~\cite{hot} has these properties, and hence we employ it for our communication complexity calculations. In addition, we assume the system has access to the following cryptographic primitives.


\noindent\textit{Digital Signatures~\cite{digital}.} We require that each node $n_i$ has a secret key $sk_i$ and a public key $pk_i$ such that 
a signature $\sigma_{sk_i}(m)$ on any message~$m$ can be created by it.
Additionally, there exists verification function $V$ such that $V(m, \sigma, pk_i) = 1$ if $\sigma = \sigma_{sk_i}(m)$, and given any~$m$ and~$pk_i$, a polynomially-bounded adversary has only negligible probability of finding $\sigma$ such that $V(m, \sigma, pk_i) = 1$ without knowledge of either $sk_i$ or $\sigma_{sk_i}(m)$.

\noindent\textit{Threshold signature~\cite{thresh}.} Each node additionally possess a secret key fragment $skt_i$, as well as a public threshold key~$pkt$ known to all nodes. Each node can then create a partial signature $\sigma_{skt_i}(m)$ over any message~$m$, such that any set $S$ of $t+1$ partial signatures over~$m$ can be combined into a single constant-sized signature $\sigma_{comb}(m, S)$. The signature can be verified with $V_t(m,\sigma, pkt)$, which will output~$1$ if 
there exists a~set $S$ of~$t+1$ partial signatures over $m$ such that $\sigma = \sigma_{comb}(m, S)$, and a polynomially bounded adversary has only a negligible probability of finding~$\sigma$ such that~$V_t(m,\sigma,pkt)=1$ without knowledge of at least~$t+1$ private keys, or the combined signature. For our purposes, we use $t = N/2$. 

\noindent\textit{Polynomial hashing~\cite{hash}.} We require a multi-variate polynomial hash function $HASH$ of low constant degree $d_{HASH}$. This polynomial then allows us to compute the hash of coded data through coded computation (see Section~\ref{subsec:setup}).

\noindent\textit{Vector commitment scheme~\cite{first_vector}.} This scheme includes the three functions $COM$, $PROVE$, and $VER$. $COM(v)$ produces a commitment $C$ to a vector $v$ of $k$ messages $m_1, m_2, ..., m_k$. $PROVE(C, i, m)$ produces a proof $\pi_i$ that~$C$ is a commitment to a vector with $m_i=m$ 
if and only if that was true. $VER(C, m, i, \pi)$ verifies that the proof is valid. The output of both $COM$ and $PROVE$ is of constant size as in \cite{pointproofs}, regardless of the the number of messages~$k$.


\subsection{System Setup} \label{subsec:setup}
\noindent Prior to every round of our protocol, the adversary is able to selectively corrupt and completely coordinate any~$f$ out of the~$N$ nodes as long as~$N\ge3f+1+(K-1)d_{HASH}$. Note the adversary is not adaptive; they can only corrupt immediately prior to a round, which we deem reasonable as no coded blockchains thus far are resilient to an adaptive adversary. Additionally, some coded blockchain systems (e.g.,~\cite{coded_journal}) have more stringent requirements, which subsume~$N\ge3f+1+(K-1)d_{HASH}$.

Blockchain systems which employ coded computation for transaction verification (coded blockchains) are a rather new idea, and the terminology and system structure is still unsettled. To apply our results to the different settings of coded blockchains, and to future implementations thereof, we devise the following canonical definition, which encompasses all coded blockchain systems known to the authors.

\begin{itemize}
    \setlength\itemsep{.5em}
    \item \textit{Block structure.} Each block~$B$ is an ordered list of transactions: $x_1, x_2, ..., x_{g(N)}$, where~$g(N)$ is the number of transactions. A transaction~$x$ contains a sender and receiver, denoted by $SEN(x)$ and $REC(x)$ respectively. Additionally, the block is split into $K=\Theta(N)$ parts $P_1, P_2, \ldots, P_{K}$ over which the encoding is performed. 
    We consider~$g(N)=\Omega(N)$, in contrast to the constant block size assumption in the standard consensus problem~\cite{byzantine} and some current blockchains (Algorand~\cite{Algorand}, Bitcoin~\cite{Bitcoin}, etc.). However, the authors believe that this assumption is compatible with blockchain systems, as it is reasonable to assume that the demand for throughput increases at least linearly with the number of nodes.
    
    \item \textit{Linear encoding.} 
    We use a slightly more general encoding scheme than mentioned for~\cite{coded_journal} to be compatible with coded blockchains that deviate from their scheme.
    There exists a~$k\times n$ generator matrix $G$ such that~$(\Tilde{P}_1,\Tilde{P}_2,\ldots,\Tilde{P}_N)=(P_1,\ldots,P_K)G$.
    Each node~$n_i$ then receives $\Tilde{P_i}$, called coded block, i.e., a linear combination of all the $P_i$'s.
    \item \textit{Coded computation.} 
    Given a polynomial $p$, each node $n_i$ computes $p(\Tilde{P_i})$ and distributes the results to the rest of the network. Upon gathering $R$ results, a node can then perform Reed-Solomon decoding to obtain $p(P_1), \ldots, p(P_K)$ so long as $(K-1)\deg(p)+2f+1\leq R$. Alternatively, nodes do not have to perform decoding themselves; this can be done by a leader, as long as the leader can prove to the nodes that no results have been fabricated, see~\cite{coded_journal}. 
    \item \textit{Leader election mechanism.} The system contains a random leader election mechanism which elects a single leader at each round, and all nodes are aware at all times who the leader is. Additionally, we assume that this process is repeatable in order to guarantee we eventually have an honest leader.
    \item \textit{Permissioned system.} Since we require the current list of all nodes to be known and not to change within an interation of our protocol, the system must be permissioned.
\end{itemize}
Having established proper background and system structure, we are in a position to formally state the problem.


\begin{probdef*} Upon successful termination of the protocol, for all $i \in [g(N)]=\{1,2,\ldots,g(N)\}$, $SEN(x_i)$ and $REC(x_i)$ have been sent a commitment $C=COM(B)$ and a proof $\pi_i=PROVE(C,i,x_i)$ for the inclusion of the transaction~$x_i$ in the block at index~$i$, and a threshold signature~$\sigma_{comb}((C,\pi_i),S)$, where~$S$ is a set of at least~$N/2+1$ partial signatures on~$(C,\pi_i)$. 
\end{probdef*}

\section{Protocol}\label{section:protocol}
Our transaction confirmation protocol relies on adding several mechanisms on top of the coded blockchain protocol. In a nutshell, our protocol begins with electing a random committee of some small size~$\lambda$ (a security parameter), which receives the uncoded block from the leader, alongside a commitment to it which shows the inclusions of each transaction (the commitment is also sent to all nodes) as well as proofs of these inclusions, i.e., the ${\pi_i}$'s. The committee members verify the correctness of the commitment, and then send a respective ``yes'' or ``no'' vote to all nodes. If a node receives a majority of ``yes'' votes from the committee, it generates a partial signature (see Section~\ref{section:background}) on the commitment from the leader. Finally, a committee member which receives partial signatures from a majority of the nodes combines them into a threshold signature, proving the nodes' confidence in the commitment. Alongside the commitment proof from the leader, the threshold signature constitutes the proof of inclusion that is readily available to the client to verify.

Upon termination, nodes will either output ``accept'' or ``reject.'' If a node outputs ``reject,'' it should behave as if that round of the coded blockchain protocol failed (i.e. not add the coded block parts it received to its chain) and a new round should begin. On output ``accept,'' the node should behave as if the round succeeded.


To combat a malicious leader, the committee uses coded computation of a \textit{polynomial hash function} to receive hashes of the uncoded parts of the block from all nodes, hence making it effectively impossible for a malicious leader to send a different block to the committee and to the nodes. The detailed steps of the protocol are as follows.


\noindent\textit{Step 1: Random committee selection.} After the leader sends out all the coded parts, we re-purpose the leader election mechanism and use it repeatedly to randomly determine a committee of~$\lambda$ members, all of whom are known to all nodes, where~$\lambda=O(1)$ is a security parameter; a higher $\lambda$ will result in higher security at the cost of higher communication complexity, and vice versa. 
 Alternatively, committee selection can be made using Algorand's sortitioning method, and the details are given in~\cite{Algorand}.
 \begin{remark}
 The following Steps 2-4 should occur simultaneously. If the leader has multiple pending tasks, it is up to their discretion as to the order they complete them in.
 \end{remark}
        
 \noindent\textit{Step 2: Block commitment.} The leader commits to $B$ via $C = COM(B)$, creates a proof~$\pi_i$ via $PROVE(C, x_i, i)$ for each transaction~$x_i$, and holds on to those for later. The leader will then broadcast the commitment~$C$ to all the nodes using Byzantine broadcast.
 
 \noindent\textit{Step 3: Uncoded reveal.} The leader shall send the entire uncoded block to each of the~$\lambda$ nodes in the random committee, along with all proofs $\pi_1,\ldots,\pi_{g(N)}$ of all transactions in the block.
 \begin{remark}
    We shall refer to the block the leader uses to generate the $\Tilde{P}_i$ as $B$ and the block sent in this step to the committee as $B'$.
\end{remark}
 \noindent The committee nodes shall verify that all proofs match the uncoded block sent to them, i.e. that $VER(C, x_i, i, \pi_i) = 1$ for all $i\in [g(N)]$.
 If all verifications in this step passed successfully, a committee node votes ``yes'' in the Committee Vote of Step~5 below, and otherwise votes ``no.''
 To ensure liveness, after $2\delta$ time following the completion of Step 1, 
 a committee member that has not received all proofs from the leader or the commitment~$C$, shall assume the leader has not sent them everything and vote ``no'' in the Committee Vote of Step~5 below.

 \begin{remark}
     It may be the case that computation of the proofs and commitment takes too long for the leader to be able to convey them to the nodes within $2\delta$ time after Step 1. In this case the $2\delta$ should simply be adjusted to some suitable agreed upon constant that is larger than $2\delta$ and allows enough time for the computations
 \end{remark}
 
 \noindent\textit{Step 4: Coded consistency check.} 
 Each node~$i$ sends $HASH(\Tilde{P_i})$ to each committee member, and then upon receiving $N-f$ hashes, each committee member uses coded computation to decode this polynomial and recover $HASH(P_i)$ for each~$i$. Therefore, committee nodes can make sure that the $P_i$'s they received match the encoded block that was distributed to every node, i.e. that $B'=B$; this is done by computing~$HASH(P_i)$ locally and making sure it is equal to the results that were decoded from the received~$HASH(\Tilde{P}_i)$. In particular, Step~3 ensured that~$C$ is a commitment to~$B'$, therefore this step ensures that~$C$ is a commitment to~$B$ by verifying that~$B'=B$. Hence, we again have the committee nodes vote ``no'' in Step 5 if any of the verifications in this step fail and ``yes'' otherwise.

\noindent\textit{Step 5: Committee vote.} 
Each committee member votes on whether~$C$ is a commitment to~$B$, i.e. that $VER(C,x_i,i,\pi_i)=1$ for every~$i$ (Step~3) and that~$B=B'$ (Step~4). 
To do this, each committee member~$n_c$ votes ``yes'' by broadcasting $\sigma_{sk_c}(1 \Vert C)$ and ``no'' by broadcasting $\sigma_{sk_c}(0 \Vert C)$ to all nodes using an additional Byzantine broadcast protocol in which the committee member serves as the leader. In the case the Byzantine broadcast protocol determines that the leader (i.e., the committee member which initiated the protocol) is malicious or non-responsive, each node considers their vote as a ``no.'' 


\noindent\textit{Step 6: Signature gathering.} First, nodes wait until the coded blockchain round terminates; if the round was accepted, they output ``accept'' and if rejected, they output ``reject.'' Since we do not want clients to need to know the committee for each round, we need to certify the commitment in a way that can be verified by only knowing the public keys of the nodes as a whole. To do this, upon receiving a majority (i.e., over~$\lambda/2$) of ``yes'' votes from the previous step, each node~$n_i$ will send $\sigma_{skt_i}(C  \Vert blockNum)$, their partial threshold signature over the commitment with the block number appended, to each member of the committee. If a node receives at least $\lceil \frac{\lambda}{2} \rceil$ ``no'' votes from the previous step, they shall instead output ``reject.'' 

\noindent\textit{Step 7: Proof distribution.} 
If a committee node receives~$N/2+1$ partial signatures over $C \Vert blockNum$ from a set of nodes~$S$, it then
uses those signatures to create $\sigma_{final} = \sigma_{comb}(C  \Vert blockNum, S)$. Each committee member shall then send out (individually) each $\pi_i$ to $SEN(x_i)$ and $REC(x_i)$ as well as both $C$ and $\sigma_{final}$.

\noindent\textit{Step 8 (optional): Transaction confirmation verification.} By this point, a client who sent or received transaction $x_i$ has received $\pi_i$, $C$, and $\sigma_{final}$. The client should then make sure that $VER(C, x_i, i, \pi_i)=1$ and $V_t(C, \sigma_{final}, pkt)=1$. Note that a client may receive multiple sets of $\pi_i$, $C$, and $\sigma_{final}$ since malicious nodes may be trying to confuse it and because each honest committee node sends a set to it. In this case, the client should only consider one set from each node for each transaction the client sent to prevent a malicious node from trying to force the client to do a large amount of excess work and should consider the transaction as included in the block if any of the considered sets passes the two verifications.

        
\begin{remark}
As an alternative to clients needing to be online to receive the proofs, we anticipate archival nodes would become common. These nodes would receive all of the proofs and store them, so that clients could query them at any point. Notice that whether or not the archival node is malicious, they have no way of corrupting the proofs. As long as the client verifies the proofs, archival nodes do not harm the integrity of the protocol.
\end{remark}

\section{Analysis}\label{section:analysis}
In this section we analyze the guarantees and communication overhead of the scheme described in the previous section. In particular, it is shown that the communication complexity of the protocol is~$O(|B|+H+N\operatorname{polylog}(N))$, and that the protocol guarantees liveness, safety, and security with high probability.
\begin{remark}
We use the bit complexity of HotStuff as an upper bound. Since it is structured as a PBFT, we are interested in the communication complexity of a \textit{view change}. Per~\cite{abraham2018hot}, the bit complexity of each view change is $O(N\operatorname{polylog}(N))$, hence we consider the bit complexity of our broadcast mechanism as $O(N\operatorname{polylog}(N))$.
\end{remark}
\begin{claim}\label{claim:communicationComplexity}
    If the original coded chain protocol has (with high probability) communication complexity $O(H)$ per block verification, the above protocol takes $O(|B|+H+N\operatorname{polylog}(N))$.
\end{claim}

\begin{proof}
    We consider the protocol step by step.
    \begin{itemize}
        \item \textit{Step 1: Random committee selection}. In this step, we use the random election system of the base chain only a constant number of times, hence the complexity of this step must be at most a constant multiple of the complexity of the base chain, i.e., $O(H(N))$.
        \item \textit{Step 2: Block commitment}. The only messages in this step are from the broadcast of the commitment, so we get $O(N\operatorname{polylog}(N))$ complexity.
        \item \textit{Step 3: Uncoded reveal}. Here the entire block and the corresponding proofs of each entry are sent by the leader to the committee. Since the committee size is constant and the individual proofs are each of constant size, the block size dominates the complexity, resulting in $O(|B|)$ complexity for this step.
        \item \textit{Step 4: Coded consistency check}. Here, each node sends a single hash to a constant number of nodes, hence we only have $O(N)$ complexity from this step.
        \item \textit{Step 5: Committee vote}. This is a constant number of constant-sized broadcasts, so we have complexity $O(N\operatorname{polylog}(N))$.
        \item \textit{Step 6: Signature gathering}. Each node sends only a constant number of signatures, leading to $O(N)$ complexity overall.
        \item \textit{Step 7: Proof distribution}. This step consists of each committee node sending at most $g(N)$ triplets of a threshold signature $\sigma_{final}$, the commitment $C$, and a single entry's proof $\pi$. Since all three parts of the triplet are constant sized, we get complexity $O(g(N))=O(|B|)$.
    \end{itemize}
    We therefore overall have bit complexity $O(|B|+H(N)+N\operatorname{polylog}(N))$
\end{proof}

\begin{remark}\label{remark:probbound}
    Since ~$N\ge3f+1+(K-1)d_{HASH}$, it trivially follows that over $\frac{2}{3}$ of the nodes are honest, so the probability of a randomly selected committee of size $\lambda$ being majority Byzantine is bounded by
    \begin{align}\label{equation:maliciouscommittee}
        \sum_{i=\lfloor \frac{\lambda}{2} + 1\rfloor}^{\lambda} {\lambda \choose i} \cdot \frac{1}{3}^i\cdot\frac{2}{3}^{\lambda-i}
    \end{align}
    For instance, with $\lambda=3000$, Eq.~\eqref{equation:maliciouscommittee} evaluates to less than $2^{-256}$. Therefore, in proving the following claims, we assume this event not to occur.
\end{remark}

\begin{claim}[Safety] \label{claim:safety}
Upon termination, either all honest nodes accept or all honest nodes reject assuming the coded blockchain possesses safety per round.
\end{claim}

\begin{proof}
    This a direct result from Step~5 and Step~6 since all nodes receive the same messages from Step~5 due to the Byzantine broadcasting, and hence in Step~6 when nodes decide to accept or reject, they must be in agreement, assuming they all agree on the outcome of the coded blockchain protocol which is guaranteed by its assumed safety. 
\end{proof}

\begin{claim}[Security]\label{claim:security} Upon termination with ``accept'' as the result, all senders of transactions have received w.h.p a proof that their transaction is in the block.
\end{claim}

\begin{proof}
For greater generality, we avoid making any assumptions regarding the nature of the communication between clients and nodes, and leave that as an implementation detail. Instead, we guarantee that a verifiable proof can be generated, and assume that a client can obtain it at will from any node once the protocol has ended. Specific mechanisms for the transmission of the proof to the client are beyond the scope of this paper, and yet these can rely on incentive mechanisms between clients and nodes.

Regardless, if the nodes accept the block, it follows that over half of the committee signed on it.

By Remark~\ref{remark:probbound}, the committee is not majority Byzantine, so at least one honest node received valid proofs relative to $C$ and $B'$ and found the hashes they decoded to be consistent with the values in the commitment. Recall that in order to correctly decode $R$ coded results, we require $(K-1)\deg(p)+2f+1\leq R$, but since $N\ge3f+1+(K-1)d_{HASH}$ and since we have at least $N-f\geq (K-1)\deg(p)+2f+1$ hash results, decoding is guaranteed to be correct. Hence, the proofs must also be valid relative to $B$. Note that the ``Proof Distribution'' step only relies on a single honest node from the committee who has the proper information. Therefore, our security property is guaranteed.
\end{proof}

\begin{claim}[Liveness] After GST, the protocol terminates within a bounded amount of time with very high probability. Additionally, the number of iterations of the protocol for all honest nodes to output "accept" is bounded assuming the coded blockchain protocol guarantees liveness.
\end{claim}
\begin{proof}
According to the properties of our broadcast mechanism, we are guaranteed that each broadcast terminates within some finite amount of time after GST. While there are spots where the leader can stall the protocol prior to GST, after GST we simply reelect the leader and within a bounded number of iterations get an honest leader. 

We begin by using the random election mechanism from the core coded chain protocol, which we assumed to have liveness. The leader then has the option to stall in Steps~2 and~3, but because honest committee nodes vote no in step 5 if they do not receive information in a timely manner in these steps, we will simply end up electing a new leader within a finite time.

Then, we have liveness over Steps~4 and~5, since we are guaranteed that the committee nodes receive at least $N-f$ hashes and via Remark~\ref{remark:probbound} the committee is, with high probability, not majority Byzantine, so at least half of the committee sends their votes out without stalling. This ensures that each node receives either a majority of ``yes'' votes or $\lceil \frac{\lambda}{2} \rceil$ ``no'' votes from the committee within some bounded amount of time, and therefore they can send their partial signatures promptly.

The final step deals with interactions with clients, which is optional, and hence is not included under the liveness property.
However, for any clients connected to the network within a finite amount of time after the final step ends, we do have liveness. Therefore, liveness is guaranteed in every step, and our entire protocol guarantees liveness after GST. 
\end{proof}


\section{Discussion}\label{section:discussion}
We point out several additional benefits of our framework in resolving other potential issues in coded blockchain systems. For instance, one approach to blockchains is an account-based model, in which an account is stored for each client, and an account-balance check is conducted whenever required. 

In account-based coded blockchain, one would have to store the accounts in a coded manner, potentially update them each round, and then employ some polynomial to verify sufficient balance. However, this only verifies that a client has sufficient balance for each individual transaction, 
rather than for all of them. In a coded scenario, this poses a significant challenge to overcome. 

Our protocol could easily solve this problem, intuitively, by gaining access to the uncoded block. We would require the leader to additionally send out a coded version of a restructured block where all transactions from a given sender are combined into one transaction with arbitrary receiver. This would allow the base chain protocol to verify that each sender has the necessary collective balance, while our protocol would ensure the coded restructured block was consistent with the original block.

An interesting direction for future research is the dual problem of notifying clients whose transactions were \textit{not} included in the block. This would require a formal model of how the clients interact with the chain and propose transactions, and would be somewhat less general. 

\printbibliography
\end{document}